\documentclass[letterpaper, 10 pt, conference]{ieeeconf}
\usepackage[textsize=scriptsize]{todonotes}
\IEEEoverridecommandlockouts
\usepackage{cite}
\usepackage{amsmath,amssymb,amsfonts}
\usepackage{algorithmic}
\usepackage{graphicx}
\usepackage{textcomp}
\usepackage{xcolor}
\usepackage[normalem]{ulem}
\usepackage{caption}
\let\labelindent\relax
\usepackage{enumitem}
\usepackage{balance}

\usepackage[font={small}]{caption}

\makeatletter

\let\proof\@undefined
\let\endproof\@undefined
\makeatother

\usepackage{amsthm,bm}
\usepackage{bbm}
\usepackage{subcaption}
\usepackage{wasysym}

\usepackage{hyperref}
\usepackage[capitalize,noabbrev]{cleveref}


\usepackage{booktabs}

\theoremstyle{plain}
\newtheorem{theorem}{Theorem}
\newtheorem{proposition}{Proposition}
\newtheorem{lemma}{Lemma}

\theoremstyle{definition}

\newtheorem{assumption}{Assumption}
\theoremstyle{remark}
\newtheorem{remark}{Remark}
\newtheorem{problem}{Problem}

\DeclareMathOperator*{\argmin}{arg\,min}
\title{\LARGE \bf
Learning Safe-by-Design Neural Network Controllers
}

\author{Yang Zhao, Jungeun Lee, Jeong hwan Jeon, Sze Zheng Yong%
\thanks{
    Y. Zhao and S.Z. Yong are with Northeastern University, Boston, USA (email: {\tt \{zhao.yang12, s.yong\}@northeastern.edu}). J. Lee and J. Jeon are with Ulsan National Institute of Science and Technology, Ulsan, Republic of Korea (email: {\tt \{jungeunlee, jhjeon\}@unist.ac.kr}).
}%
\thanks{
    This work was supported by the NSF grant CNS-2313814, and by the InnoCORE program of the Ministry of Science and ICT (\#N10250155).
}%
}

\begin{document}

\maketitle
\thispagestyle{empty}
\pagestyle{empty}

\begin{abstract}
Safety filters constructed from control barrier functions (CBFs) are commonly appended to pre-trained neural network controllers to enforce safety requirements. However, this decoupled design with
hand-tuned, fixed CBF parameters often fails to adapt to the underlying controller, yielding overly conservative solutions. Thus, given a valid CBF, we address these limitations by jointly learning a neural network controller and neural-network-parameterized CBF parameters,
enforcing the resulting affine safety constraints by construction and avoiding an online quadratic program (QP) safety filter at run time. To further improve computational efficiency and scalability, we introduce a lightweight projection architecture that enforces constraints without full constraint enumeration. Extensive simulation evaluations
demonstrate reliable, scalable safety constraint satisfaction at reduced computational cost.

\end{abstract}

\section{Introduction}

Neural networks are increasingly used as powerful function approximators for feedback control policies, enabling controllers to handle complex nonlinear dynamics and high-dimensional state spaces that are difficult to model analytically \cite{hunt1992neural,nakamura2022neural}.
However, in safety-critical systems, purely learned controllers may produce unsafe control actions that violate system constraints such as collision avoidance, actuator limits, or state safety requirements.
A common strategy to address this is to place a safety filter after the neural network controller, as shown in Fig.~\ref{fig:safety_filter} (left).
While this modular design allows learned controllers to be combined with formal safety guarantees, it may lead to
suboptimal behavior.

Control barrier functions (CBFs) have become one of the most widely used tools for implementing such safety filters in the process of designing control inputs \cite{ames2017control}. In the standard CBF framework, safety is enforced by imposing a constraint of the form $\dot h(x,u) + \alpha(h(x)) \ge 0$, with

a CBF candidate $h(\cdot)$ and an extended class-$\mathcal{K}$ function $\alpha(\cdot)$ that regulates how aggressively the system is pushed back toward the safe set. In practice, however, the design of this $\alpha(\cdot)$ function is often heuristic,
e.g., by fixing
it as a simple linear function (e.g., $\alpha(h)=\omega h$) with a manually tuned $\omega$, which can lead to conservative behavior or unsafe motion \cite{kim2025adapt}.
Optimal-decay CBF methods address this limitation by introducing an adaptive decay variable into the online CBF-QP, improving flexibility over fixed-gain designs~\cite{ong2025properties}. However, they still require solving a (larger) QP at run time.
BarrierNet also learns higher-order CBF components for safe robot control, but the training process is complex and expensive, and it retains a CBF-QP layer.
By contrast, our goal is to embed the resulting affine safety constraints directly into the neural controller~\cite{xiao2023barriernet} with optimal-decay, without any online QP.

Another line of work enforces safety constraints directly within neural network controllers.
Penalty-based methods promote the generation of safe control actions during training, but typically do not provide hard safety guarantees, especially under limited data or out-of-distribution conditions  \cite{marquez2017imposing}.
To guarantee safety, projection-based methods
have been proposed to enforce hard constraints by projecting the neural network output onto a feasible set defined by safety constraints.
HardNet~\cite{min2025hardnet}
enforces  state-dependent affine constraints by adding correction terms%
, but requires $A(x)$ to have full row rank, limiting the number of constraints and failing to handle linearly dependent cases or non-unique projections.
More recently, constraint-affine neural networks (CAffNet) \cite{zhao2026caffnet} addressed these limitations by incorporating a constraint-affine layer with constraint decomposition and a trainable projection module, enabling arbitrary state-dependent affine constraints with provable feasibility and universal approximation properties.
However, the projection step can be computationally expensive as the number of constraints increases, limiting its scalability to complex systems.

\begin{figure}[t]
    \centering
    \includegraphics[width=1\linewidth]{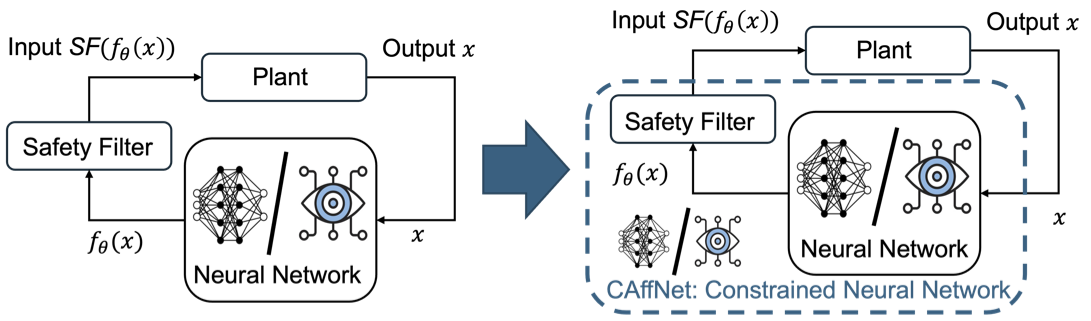}%
    \caption[]{\small Framework of the trainable safety filter. The proposed method jointly trains a neural network controller and CBF parameters within the CAffNet(-Lite) architecture, ensuring hard constraint satisfaction by construction without a separate safety filtering step.}\vspace{-0.2cm}
    \label{fig:safety_filter}
\end{figure}

Motivated by these challenges, we propose a safe-by-design neural network control framework that integrates trainable CBFs with CAffNet(-Lite) in Fig.~\ref{fig:safety_filter} (right).
The main contributions are summarized as follows:
\begin{enumerate}[wide, labelwidth=!, labelindent=0pt]
\item We introduce CAffNet-Lite, an extension of CAffNet with a lightweight projection architecture that reduces computation time while preserving hard affine constraint satisfaction and
its universal approximation property.
\item Given a valid CBF under input constraints, we leverage CAffNet-Lite to train a neural network controller that satisfies an optimal-decay-like property by construction, enabling real-time deployment without the computational overhead and reliability concerns of solving a QP online.
\item We prove that the resulting piecewise continuous control satisfies strong invariance of the CBF in the Filippov sense.
\end{enumerate}

\section{Modeling and Problem Statement}

Consider a nonlinear control-affine system
\begin{align}
\dot{x} &= f(x) + g(x)u,
\label{eq:dynamics}
\end{align}
where $x \in \mathbb{R}^{n}$ denotes the system state and $u \in \mathcal{U} \subseteq \mathbb{R}^{m}$ the control input, with a polytope input set,
i.e., $\mathcal{U} = \{u \in \mathbb{R}^{m} \mid P u \le q\}$.
The functions $f:\mathbb{R}^{n} \rightarrow \mathbb{R}^{n}$ and $g:\mathbb{R}^{n} \rightarrow \mathbb{R}^{n \times m}$ are assumed to be locally Lipschitz continuous.

To characterize system safety, let $h:\mathbb{R}^{n} \rightarrow \mathbb{R}$ be a continuously differentiable function that defines the safe set
\begin{align}
\mathcal{C}
&=
\{x \in \mathbb{R}^{n} \mid h(x) \ge 0\},
\label{eq:safeset}
\end{align}
where $\mathcal{C}$ satisfies some desired state constraints $\mathcal{X}$, i.e., $\mathcal{C} \subseteq \mathcal{X}$, e.g., for collision avoidance.
The objective is to ensure forward invariance of $\mathcal{C}$, meaning that if the system starts inside the safe set, it remains inside for all future time.
A common approach to enforce safety is through control barrier functions (CBFs). A function $h(x)$ is a CBF if there exists an extended class-$\mathcal{K}$ function $\alpha(\cdot)$ such that for all $x \in \mathbb{R}^n$:
\begin{align}
\sup_{u \in \mathcal{U}}
\left[
L_f h(x) + L_g h(x)u
\right]
\ge -\alpha(h(x)),
\label{eq:CBF_condition}
\end{align}
where $L_f h(x) = \tfrac{\partial h(x)}{\partial x} f(x)$ and $L_g h(x) = \tfrac{\partial h(x)}{\partial x} g(x)$ are the Lie derivatives.
Condition \eqref{eq:CBF_condition} ensures that the safe set $\mathcal{C}$ is forward invariant under an appropriate control input.
In practice, the overall (safety) CBF constraints are often enforced by modifying a nominal control input $u_{\mathrm{nom}}(x)$ through a quadratic program (QP):
\begin{align}
\label{eq:CBF_qp}
\begin{aligned}
u^*(x)
&=
\argmin_{u \in \mathcal{U}}
\|u - u_{\mathrm{nom}}(x)\|^2_2 \\
& \text{s.t.}\quad
\begin{aligned}
     -L_g h(x)u &\le L_f h(x) + \alpha(h(x)).
\end{aligned}
\end{aligned}
\end{align}

Since the CBF condition is affine in
$u$ and the input set $\mathcal{U}$ is polytopic by assumption, the constraints in \eqref{eq:CBF_qp} can be written equivalently as
state-dependent affine constraints:
\begin{align}\label{eq:affine_constraints}
    \mathcal{S}(x) := \{u \mid A(x)u \le b(x)\},
\end{align}
with $A(x)=\begin{bmatrix}-L_gh(x)\\P\end{bmatrix}$ and $b(x)=\begin{bmatrix}L_f h(x) + \alpha(h(x))\\q\end{bmatrix}$.

The QP safety filter in \eqref{eq:CBF_qp} provides a principled mechanism to enforce safety constraints, but solving the optimization problem online can introduce computational overhead and may degrade control performance due to frequent corrections to the nominal controller. Moreover, manually choosing and fixing the CBF parameters, i.e., the $\alpha(\cdot)$ function, may diminish the performance of the nominal controller.

Thus,
this letter aims to jointly design a neural network-based controller and neural-network-parameterized CBF parameters (similar to optimal-decay CBFs \cite{ong2025properties}) that satisfy the CBF safety constraints by construction without requiring an online QP safety filter. Formally, the problem is as follows:
\begin{problem}[Safe-by-Design Neural Network Control]
\label{prob:safe_control}
Consider the control-affine system in \eqref{eq:dynamics} with state $x \in \mathbb{R}^{n}$ and control input $u \in \mathcal{U} \subseteq \mathbb{R}^{m}$. Let the safe set be defined by \eqref{eq:safeset} with a given valid/non-empty input-constrained CBF.
The objective is to design a neural network controller $u = \pi_\theta(x) \in \mathcal{U} $ together with trainable CBF parameters as a neural network $\alpha_\vartheta(x,h(x))$ replacing the function $\alpha(h(x))$ in  \eqref{eq:CBF_qp} such that the safe set $\mathcal{C}$ is (strongly) forward invariant.
\end{problem}

To solve this problem, we develop a safe-by-design neural network controller based on a novel lightweight CAffNet-Lite that extends the CAffNet architecture in our prior work \cite{zhao2026caffnet}, with a trainable control affine constraint.

\vspace{-0.15cm}
\section{Main Approach}
\vspace{-0.05cm}

In this section, we first introduce CAffNet-Lite in Section \ref{sec:CAffNet-Lite}, an improved safe-by-design neural network control framework inspired by~\cite{zhao2026caffnet}, that integrates a constraint-affine
layer with neural networks (NNs) to guarantee satisfaction of state-dependent control affine constraints in \eqref{eq:affine_constraints} by construction, which we employ to solve Problem \ref{prob:safe_control} in Section \ref{sec:sNNcontroller}. Moreover, we prove  its capability to enforce  strong forward invariance of a given safe set in the Filippov sense, and thus, to address Problem~\ref{prob:safe_control}.

\vspace{-0.15cm}
\subsection{Constrained Neural Network Design via CAffNet-Lite} \vspace{-0.05cm}\label{sec:CAffNet-Lite}
The constrained neural network framework (cf. Fig. \ref{fig:safety_filter} (right)) has the state measurement $x$ as an input $x$ and generates the feedback control input $u$ as its output, and
can be parameterized by any trainable universal function approximator $f_\theta(x)$ establishing a mapping from state $x$ to control $u$, including common architectures such as feedforward neural networks or transformers \cite{park2021minimum, yun2019transformers}. Specifically, the proposed CAffNet-Lite extends our prior affine-constrained neural network architecture in \cite{zhao2026caffnet} by improving its scalability with an improved constraint decomposition strategy, which will be introduced in Section~\ref{sec:cons_decomp}.
Then, in Section~\ref{sec:projection}, we present the projection architecture, inherited and recapped from CAffNet \cite{zhao2026caffnet}, that integrates the unconstrained output, the null-space component, and the decomposed constraints to guarantee hard constraint satisfaction and feasibility of the resulting control policy.

\subsubsection{Lightweight Constraint Decomposition}\label{sec:cons_decomp}

Given that \eqref{eq:affine_constraints} with $n_c$ number of constraints forms a polytope, its boundary is defined by a maximum $m$ number of active constraints. Then, as with CAffNet \cite{zhao2026caffnet}, we  enforce the constraints defined by~\eqref{eq:affine_constraints}  by  decomposing $A(x)$, $b(x)$ into subconstraints. However, instead of the CAffNet approach \cite{zhao2026caffnet} that selects all possible combinations of active constraints up to order $\min(n_c, m)$, we propose CAffNet-Lite that only considers
subconstraints with either the minimum or maximum order, i.e.,
$k \in \{1, \min{(n_c, m)}\}$,
with a (reduced) set of indices:
\begin{align*}
    {\Gamma} := \{
        \gamma
        \mid
        k \in \bigl\{1, \min{(n_c, m)} \bigr\}, 1 \leq j_1 < \dots < j_k \leq n_c
    \},
\end{align*}
where $\gamma = (j_1, \dots, j_k)$ is the index sequence, and the original constraint is decomposed into subconstraints with each $\gamma \in \Gamma$,
i.e., with
the corresponding submatrix of $A(x)$ and subvector of $b(x)$ with rows $a_{j_i}^\top$ and $b_{j_i}$ selected by $\gamma$ as $A_{\gamma} =
\begin{bmatrix}
a_{j_1} & \hdots & a_{j_k}
\end{bmatrix}^\top$ and $b_{\gamma} =
\begin{bmatrix}
b_{j_1} &  \hdots & b_{j_k}
\end{bmatrix}^\top$.
Note that we do not require $A_\gamma$ to have full row rank.

Under this new design, the total number of considered constraint subsets is $n_c$ when $\min(n_c,m)=1$, and otherwise $n_c + \tbinom{n_c}{\min(n_c,m)}$, which yields a worst-case polynomial complexity of $O(n_c^{m})$ with respect to $n_c$. By contrast, CAffNet \cite{zhao2026caffnet}  that selects all possible combinations of active constraints yields $\sum_{k=1}^{\min(n_c,m)} \binom{n_c}{k}$ subconstraints, which is upper bounded by $2^{n_c}-1$ and thus, has worst-case exponential complexity $O(2^{n_c})$ with respect to $n_c$; consequently, the number of projection candidates for CAffNet is reduced
for systems with many safety constraints or high-dimensional control inputs, leading to significantly increased scalability.

\subsubsection{Projection with Sub-Constraints} \label{sec:projection}

The projection architecture is adopted from CAffNet \cite{zhao2026caffnet}, under the following assumption on the continuity and feasibility of the constraints.
\begin{assumption}\label{assumption:cont_feas_const}
For any
state $x \in \mathbb{R}^{n}$, $A(x)$ and $b(x)$ are continuous in $x$, and the feasible region $\mathcal{S}(x) := \{y \in \mathbb{R}^{m} \mid A(x) y \le b(x)\}$ is non-empty.
\end{assumption}
With sub-constraints ($A_\gamma(x),b_\gamma(x))$ from Section~\ref{sec:cons_decomp}, the projected output is given by:
\begin{align}
\begin{array}{rl}
\mathcal{P}_\gamma (x)
= \!\!\!
& f_{\theta}(x) - A_{\gamma}^\dagger(x) (A_{\gamma}(x) f_{\theta}(x) - b_{\gamma}(x)) \\
&+ (I - A_{\gamma}^\dagger (x) A_{\gamma} (x)) w_\phi(x),
\end{array}\label{eq:projSub}
\end{align}
where $A_{\gamma}^\dagger(x)$ is the pseudoinverse of $A_{\gamma}(x)$,
$f_\theta(x) \in \mathbb{R}^{m}$ is the output of an unconstrained NN, and $w_\phi(x) \in \mathbb{R}^{m}$ is a learned null-space component
that allows for selecting different feasible points on the hyperplane defined by the sub-constraint when the solution of $A_{\gamma}(x)y = b_{\gamma}(x)$ is not unique (instead of only allowing orthogonal/parallel projections).
Then, the feasible candidate set is  defined by selecting all projections that are inside the safe region, i.e.,
\begin{align}\label{eq:feasible_proj_set_lite}
    \mathcal{S}_{\mathcal{P}}(x) :=
    \left\{
    \mathcal{P}_{\gamma}(x)
    \,\middle|\,
    \gamma \in \Gamma,\;
    A(x)\mathcal{P}_{\gamma}(x) \le b(x)
    \right\}.
\end{align}
Finally, CAffNet-Lite outputs the unconstrained network $f_\theta(x)$ when feasible; otherwise it selects the projection from \eqref{eq:feasible_proj_set_lite} that is closest to $f_\theta(x)$ under any $p$-norm ($\|\cdot\|_p$), i.e.,
\begin{align}\label{eq:y_star_lite}
    {\mathcal{P}}^*(x)
    \,\in
    \begin{cases}
        f_\theta(x), & \text{if } f_\theta(x) \in \mathcal{S}(x), \\
        \displaystyle\arg\!\!\!\! \min_{y \in  {\mathcal{S}}_{\mathcal{P}} (x)} \! \|y - f_\theta(x)\|_p, & \text{otherwise}.
    \end{cases}\hspace{-0.1cm}
\end{align}
Note that in the event of a tie, the tie can be arbitrarily broken. Moreover, CAffNet-Lite does not require full-rank or linearly independent constraints.

The projected output with the improved constraints decomposition in Section \ref{sec:cons_decomp} is
guaranteed to be feasible.

\begin{lemma}[Existence of a Feasible Solution]\label{lemma:feasible_sol_proj}
For any
state $x \in \mathbb{R}^{n}$, if the feasible region $\mathcal{S}(x) := \{y \in \mathbb{R}^{m} \mid A(x)y \le b(x)\}$ is non-empty, then there exists at least one candidate
$y \in \{\mathcal{P}_{\gamma}(x) \mid \gamma \in  {\Gamma}\}$ such that $A(x)y \le b(x)$.
\end{lemma}\vspace{-0.2cm}

\begin{proof}
    Based on the definition of the minimal face~\cite[Theorem 8.4]{schrijver1998linear} and the proof of feasibility in CAffNet~\cite[Appendix A]{zhao2026caffnet},
    the projection onto the minimal face of the polyhedron (with maximum order $k = \min{(n_c, m)}$) is, by definition, within the feasible region.
    Such index set satisfies $ {\Gamma}_k = \{\gamma \mid k = \min(n_c, m)\} \subseteq  {\Gamma}$ by construction. Therefore, the corresponding projection set $\{ \mathcal{P}_{\gamma}(x) \mid \gamma \in  {\Gamma}_{\min(n_c, m)} \}$ is non-empty. Thus, with the definition of the minimal face, this projection set satisfies the constraints in~\eqref{eq:affine_constraints}, which indicates~\eqref{eq:feasible_proj_set_lite} is non-empty and Lemma~\ref{lemma:feasible_sol_proj} holds.%
\end{proof}\vspace{-0.1cm}

Then, since all solutions of \eqref{eq:y_star_lite}  satisfy $\mathcal{S}(x)$, which exist by Lemma~\ref{lemma:feasible_sol_proj}, the following holds directly.\vspace{-0.1cm}

\begin{proposition}[Hard Constraint Satisfaction] \label{theorem:const_sat_of_CAffNet}
For any
state $x \in \mathbb{R}^{n}$, under Assumption~\ref{assumption:cont_feas_const}, the output of CAffNet-Lite $\mathcal{P}^*(x)$ in~\eqref{eq:y_star_lite} satisfies the state-dependent affine constraints, i.e., $A(x) \mathcal{P}^*(x) \le b(x)$.
\end{proposition}

Furthermore, it can be shown that CAffNet-Lite preserves the universal approximation property of CAffNet.\vspace{-0.1cm}

\begin{proposition}[Performance Preservation]\label{lemma:performance_preservation}
Let $y^*(x)$ be an optimal output of CAffNet (that has the universal approximation property \cite[Theorem 3.5]{zhao2026caffnet}).
CAffNet-Lite in~\eqref{eq:y_star_lite} can approximate $y^*(x)$ arbitrarily well.
\end{proposition}
\begin{proof}
If $y^*(x)=f_\theta(x)$, then $f_\theta(x)\in \mathcal{S}(x)$ by \eqref{eq:y_star_lite}, and CAffNet-Lite returns the same output as CAffNet. Otherwise, suppose $y^*(x)=\mathcal{P}_{\gamma^*}(x)$. If $\gamma^*\in {\Gamma}$, the same candidate is already included in CAffNet-Lite. Thus, it remains to consider an intermediate-order subset $\gamma^*\in\Gamma^o\setminus {\Gamma}$, where $\Gamma^o$ is the set of indices in CAffNet \cite[above Eq. (2)]{zhao2026caffnet}. Since $y^*(x)$ satisfies $A_{\gamma^*}(x)y^*(x)=b_{\gamma^*}(x)$, there exists a single-constraint index $\tilde{\gamma}^* = (j^*)$, where $j^*$ is an entry of $\gamma^*$, such that $a_{j^*}(x)y^*(x)=b_{j^*}(x)$. In addition, there exists a target null-space function $w_t^*(x)=y^*(x)-f_\theta(x)$. By the universal approximation property of NNs, $w_\phi(x)$ can approximate $w_t^*(x)$ arbitrarily well. With $w_\phi(x)=w_t^*(x)$ and omitting the argument $x$, the CAffNet-Lite projection for the single constraint with the index $(j^*)$ recovers $y^*$, since $\mathcal{P}_{\tilde{\gamma}^*}
= f_\theta -a_{j^*}^\dagger (a_{j^*}f_\theta -b_{j^*}) + (I-a_{j^*}^\dagger a_{j^*})w_t^* = y^* - a_{j^*}^\dagger
(a_{j^*}y^* - b_{j^*})
= y^*$.
\end{proof}

\vspace{-0.3cm}
\subsection{Neural Network Based Safety Filter}\label{sec:sNNcontroller}\vspace{-0.1cm}

To overcome the limitations of a fixed-parameter class-$\mathcal{K}$ function in the standard CBF formulation, we propose to learn the CBF parameters  jointly with the neural network controller within the CAffNet-Lite architecture. Specifically, similar to optimal-decay CBFs \cite{ong2025properties}, we replace the class-$\mathcal{K}$ function $\alpha(\cdot)$ in the CBF condition in \eqref{eq:CBF_condition} with a neural network parameterization $\alpha_\vartheta(x,h(x))$ that satisfies $\alpha_\vartheta(x,0)=0$ at the safety boundary when $h=0$ to enforce invariance, yielding more flexible safety constraints that better capture the system dynamics and safety requirements.

Under this construction, we have trainable state-dependent affine constraints $A(x) u \le b_\vartheta(x)$ with:
\begin{align}
\begin{gathered}
    A(x)\!=\! \begin{bmatrix}-L_gh(x)\\P\end{bmatrix}\!,\
    b_\vartheta(x) \!=\! \begin{bmatrix}L_fh(x) \!+\! \alpha_\vartheta(x,h(x))\\q\end{bmatrix}\!,
\end{gathered}
\label{eq:CBF_affine_constraints}
\end{align}
which, for simplicity and similar to optimal-decay CBF \cite{ong2025properties}, can be imposed by choosing $\alpha_\vartheta(x,h(x)) = \tilde{\alpha}_\vartheta(x) \alpha(h(x))$ with $\tilde{\alpha}_\vartheta (x)$ being a trainable (unconstrained) neural network and the original $\alpha(\cdot)$. Future work will explore how to more intelligently impose this constraint for better performance.

Unlike the standard CBF that
relies on continuous or sufficiently regular closed-loop dynamics, the CAffNet-Lite controller in \eqref{eq:y_star_lite} may be discontinuous since the $\argmin$ selection in the projection layer can have ties and the set of constraints in $\mathcal{S}_{\mathcal{P}}$ may vary with $x$. We thus establish invariance in the Filippov sense, i.e.,
every Filippov solution satisfies the CBF inequality almost everywhere.

\begin{theorem}[Safe-by-Design Neural Network Controllers]
Let the control input be given by \eqref{eq:y_star_lite}, with the state-dependent affine constraints
in \eqref{eq:CBF_affine_constraints}.
Assume that
$\alpha_\vartheta(x,h)$ is locally Lipschitz in $h$ with $\alpha_\vartheta(x,0)=0$
and that $K_\vartheta(x):=\{u\mid A(x)u\le b_\vartheta(x)\}$ is nonempty, $\forall x\in\mathbb R^n$.
Then, the
CAffNet-Lite controller with ${\mathcal{P}}^*$ in \eqref{eq:y_star_lite} guarantees strong forward invariance of the safe set $\mathcal{C}$ and thus, solves Problem~\ref{prob:safe_control}.
\end{theorem}\vspace{-0.35cm}

\begin{proof}
    First, note that each candidate $\mathcal{P}_\gamma(x)$ is continuous (as neural networks are continuous functions). Thus, $u(x)$ is a piecewise continuous set-valued function whose domain partitions according to active constraints in $\mathcal{S}_{\mathcal{P}}(x)$. At partition boundaries or tie points, set-valuedness arises. Since $A(x), b(x)$, and each $\mathcal{P}_\gamma(x)$ are continuous, limits $\lim_{x \to x^*} u(x) = u^*$ from any partition satisfy $A(x^*)u^* \le b(x^*)$ (by continuity of matrix-vector multiplication and preservation of inequality under limits). For ties, all candidate limits satisfy the constraints identically.

Define the velocity as $X(x) = f(x) + g(x)u(x)$ and consider the Filippov regularization:
\begin{equation*}
\dot{x} \in \mathcal{F}[X](x) := \overline{\text{co}}\{f(x) + g(x)k_i(x) \mid k_i(x) \in \lim_{x \to x^*} u(x)\}.
\end{equation*}
 Filippov solutions exist locally~\cite[Prop.~3]{cortes2008discontinuous} since $X$ is piecewise continuous and locally essentially bounded (finitely many pieces, bounded $\mathcal{S}_{\mathcal{P}}(x)$).

For any Filippov solution $x(\cdot)$ with $x(0) \in \mathcal{C}$, each limiting control $k_i(x)$ satisfies the constraint $L_fh(x) + L_gh(x)k_i \ge -\alpha_\vartheta(x,h(x))$. By linearity in $u$, their convex combinations $\bar{u} = \sum_i \lambda_i k_i$ satisfy the same inequality. Since $h$ is $C^1$ and $x(\cdot)$ absolutely continuous, $h(x(t))$ satisfies the differential inequality a.e. Since $\alpha_\vartheta$ is locally Lipschitz with $\alpha_\vartheta(x,0)=0$ (hence minimal~\cite[Cor.~1]{konda2020characterizing}), the comparison lemma~\cite[Lemma~3.4]{khalil2002nonlinear} yields $h(x(t)) \ge 0$ for all $t$.
\end{proof}

\vspace{-0.2cm}
\begin{remark}[Vector-valued CBFs]
Vector-valued CBFs are interpreted componentwise: $\mathcal C=\bigcap_{j\in\mathcal J}\mathcal C_j$, where $\mathcal C_j=\{x\mid h_j(x)\ge0\}$ and $h_j$ is continuously differentiable. The affine constraint in \eqref{eq:CBF_affine_constraints} is replaced by the stacked CBF constraints. Compared with the scalar case,
we require simultaneous feasibility, $K_\vartheta(x)\neq\emptyset$, e.g., with MFCQ-type compatibility.
Under the stacked feasibility condition, the theorem applies componentwise to each $h_j$ along the same Filippov solution.
\end{remark}

\vspace{-0.2cm}
\section{Simulation Results}
\vspace{-0.05cm}
In this section, we demonstrate
the advantages of jointly learning the control policy and the safety conditions with our proposed safe-by-design neural network control framework.
Following the setup in \cite{zhao2026caffnet}, we use a penalty cost of $100\text{ReLU}(A(x)f_{net}(x)-b(x))$ as the soft constraint for the standard neural network and HardNet.
We additionally apply the proposed framework using feedforward neural networks, namely CAffNet and CAffNet-Lite with the network hyperparameters
specified in each example. For all simulations, we use the Adam optimizer with a learning rate of 0.0001 and the 2-norm in \eqref{eq:y_star_lite}. Each simulation is repeated five times with different random seeds.
All models are implemented and trained using PyTorch~\cite{paszke2019pytorch}. The training and testing are performed with NVIDIA Tesla V100-SXM2-32GB.

Performance is evaluated based on three primary metrics: total loss, constraint violations, and computational overhead. The total loss is defined as the Mean Squared Error (MSE) between the controller output and the nominal command, averaged across all training samples. Constraint satisfaction is quantified by $r = \operatorname{ReLU}(A(x) u - b(x))$, for which we report the maximum value, the mean, and the percentage of violations across the test set.
Computational efficiency is assessed via the average training time per epoch ($T_{\text{train}}$ in milliseconds (ms)) and the execution time per test set ($T_{\text{test}}$ in seconds (s)).
In the following tables, main entries represent the mean across five independent runs, while values in parentheses denote the standard deviation.

\vspace{-0.15cm}
\subsection{Single Integrator System with Collision Avoidance}
\vspace{-0.05cm}

\begin{table}[t]
  \caption{Results of learning safe control policies for a single integrator system with parenthetical values representing standard deviations.
  CAffNet and CAffNet-Lite achieve the lowest mean cost while strictly satisfying all trained constraints. Both NN and HardNet have some constraint violations. QP with fixed CBF parameters has a relatively higher cost and
  a longer testing time.}
  \label{CBF-table}
  \begin{center}
      \begin{sc}
\resizebox{\linewidth}{!}{%
\begin{tabular}{lcccccc}
\toprule
Method ($\omega$) & Cost & \multicolumn{3}{c}{Inequality violation} & $T_{train}$ & $T_{test}$ \\
\cmidrule(lr){3-5}
 &  & Max & Mean & Num (\%) & (ms) & (s) \\
\midrule
QP & 1095.09 & \textbf{0.00} & \textbf{0.00} & \textbf{0.00\%} & \textbf{0.00} & 14.46 \\
($\omega=0.1$) & (0.00) & (0.00) & (0.00) & (0.00\%) & (0.00) & (0.51) \\
OD-QP & 574.23 & \textbf{0.00} & \textbf{0.00} & \textbf{0.00\%} & \textbf{0.00} & 13.26 \\
($\omega=0.1$) & (0.00) & (0.00) & (0.00) & (0.00\%) & (0.00) & (1.30) \\
QP & 378.67 & \textbf{0.00} & \textbf{0.00} & \textbf{0.00\%} & \textbf{0.00} & 13.57 \\
($\omega=10$) & (0.00) & (0.00) & (0.00) & (0.00\%) & (0.00) & (0.23) \\
OD-QP & 378.71 & \textbf{0.00} & \textbf{0.00} & \textbf{0.00\%} & \textbf{0.00} & 14.00 \\
($\omega=10$) & (0.00) & (0.00) & (0.00) & (0.00\%) & (0.00) & (1.18) \\
NN ($\tilde{\alpha}_\vartheta$) & $1.90 \times 10^7$ & 5.59 & 0.51 & 2.67\% & 15.22 & \textbf{0.44} \\
 & $(2.24 \times 10^7)$ & (3.48) & (0.32) & (1.35\%) & (0.17) & (0.00) \\
HardNet ($\tilde{\alpha}_\vartheta$) & $5.02 \times 10^5$ & 1.13 & 0.10 & 2.97\% & 24.74 & 6.98 \\
 & $(1.86 \times 10^5)$ & (0.14) & (0.01) & (0.46\%) & (0.47) & (0.14) \\
CAffNet ($\tilde{\alpha}_\vartheta$) & \textbf{226.03} & \textbf{0.00} & \textbf{0.00} & \textbf{0.00\%} & 95.88 & 11.34 \\
 & (95.07) & (0.00) & (0.00) & (0.00\%) & (1.64) & (0.73) \\
CAffNet-Lite & \textbf{226.03} & \textbf{0.00} & \textbf{0.00} & \textbf{0.00\%} & 95.57 & 11.30 \\
($\tilde{\alpha}_\vartheta$) & (95.07) & (0.00) & (0.00) & (0.00\%) & (1.65) & (1.05) \\
\bottomrule
\end{tabular}
}
      \end{sc}
  \end{center}
\end{table}

\begin{figure}[t]
    \centering
    \begin{subfigure}[t]{0.16\textwidth}
        \centering
        \includegraphics[width=\linewidth,trim=0mm 5mm 0mm 2mm]{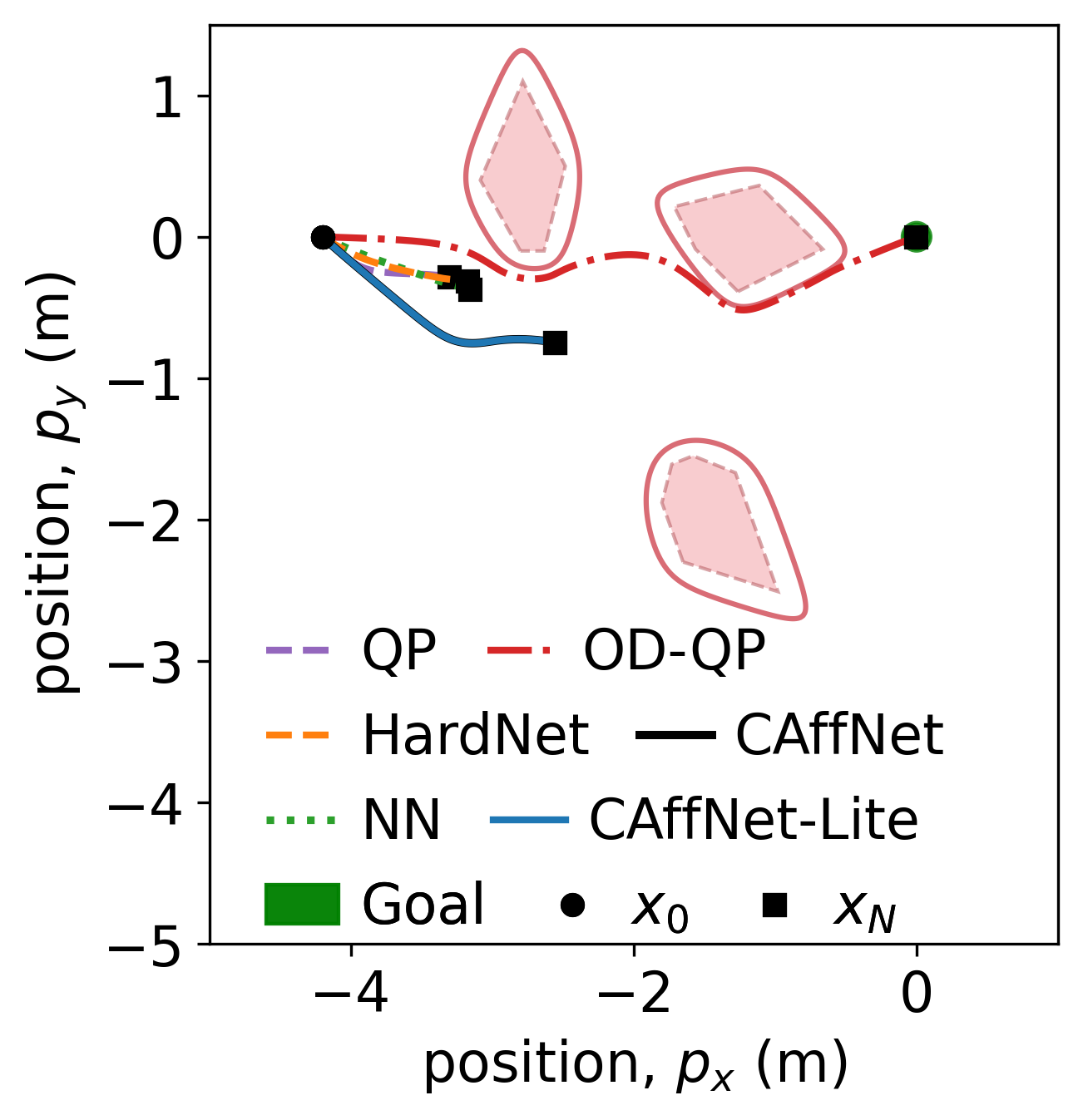}
        \caption[]{\small $\alpha(h)=0.1h$}
        \label{fig:single_int_traj_omega_0p1}
    \end{subfigure}\hfill
    \begin{subfigure}[t]{0.16\textwidth}
        \centering
        \includegraphics[width=\linewidth,trim=0mm 5mm 0mm 2mm]{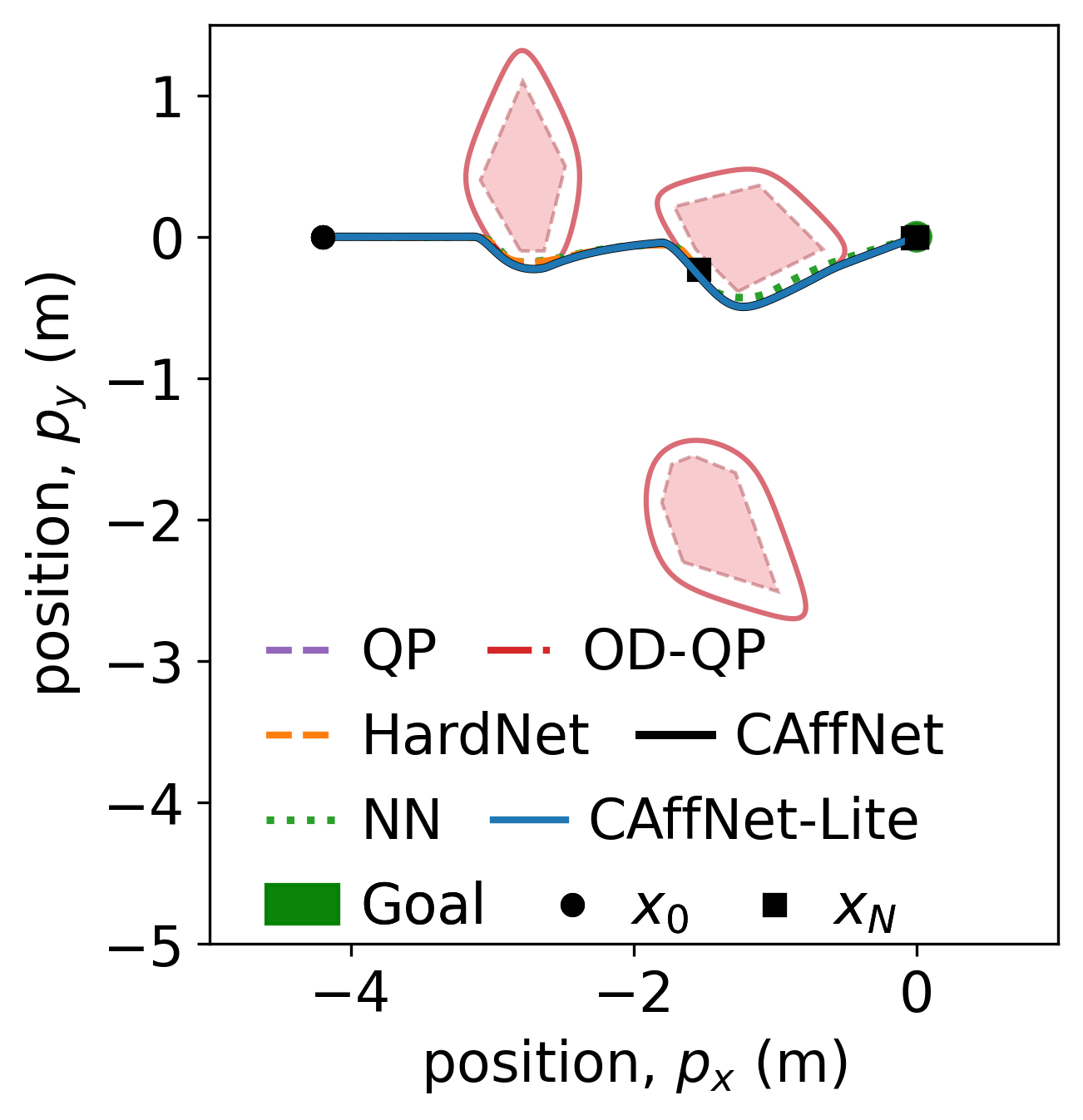}
        \caption[]{\small $\alpha(h)=10h$}
        \label{fig:single_int_traj_omega_10}
    \end{subfigure}\hfill
    \begin{subfigure}[t]{0.16\textwidth}
        \centering
        \includegraphics[width=\linewidth,trim=0mm 5mm 0mm 2mm]{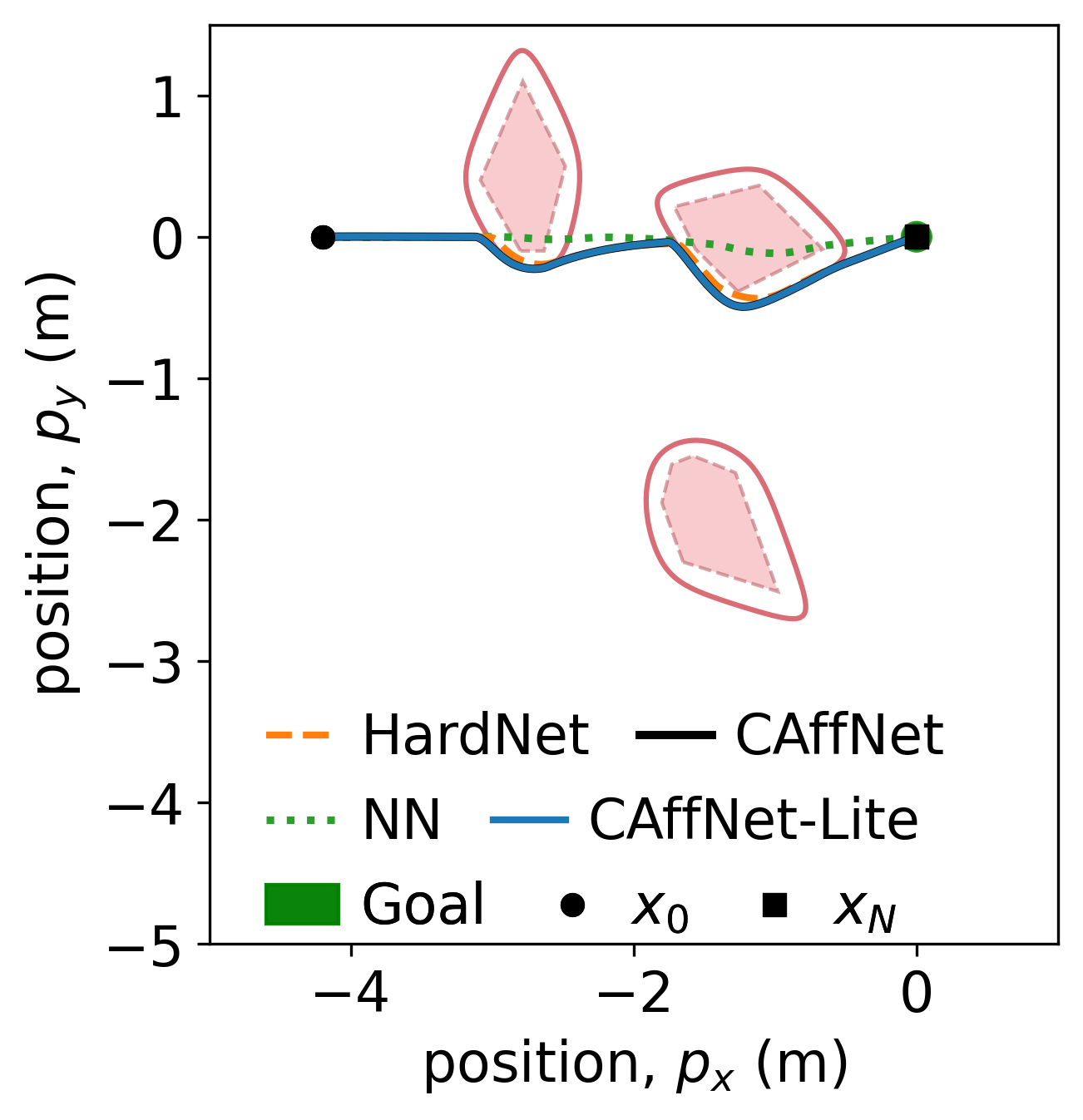}
        \caption[]{\small $\!\alpha_\vartheta(x,\!h) \!=\!   \tilde{\alpha}_\vartheta(x) h$}
        \label{fig:single_int_traj_omega_learned}
    \end{subfigure}
    \vspace{-0.15cm}
    \caption[]{Trajectories of various controllers with different fixed class-$\mathcal{K}$  functions $\alpha(h)$ or learned $\alpha_\vartheta(x,\!h)$ for a single integrator system.} \vspace{-0.2cm}%
    \label{fig:CBF_trajectory}
\end{figure}

\begin{figure*}[t]
    \centering
    \vspace{0.3cm}
    \begin{subfigure}[t]{0.14\textwidth}
        \centering
        \includegraphics[width=\linewidth,trim=0mm 5mm 0mm 0mm]{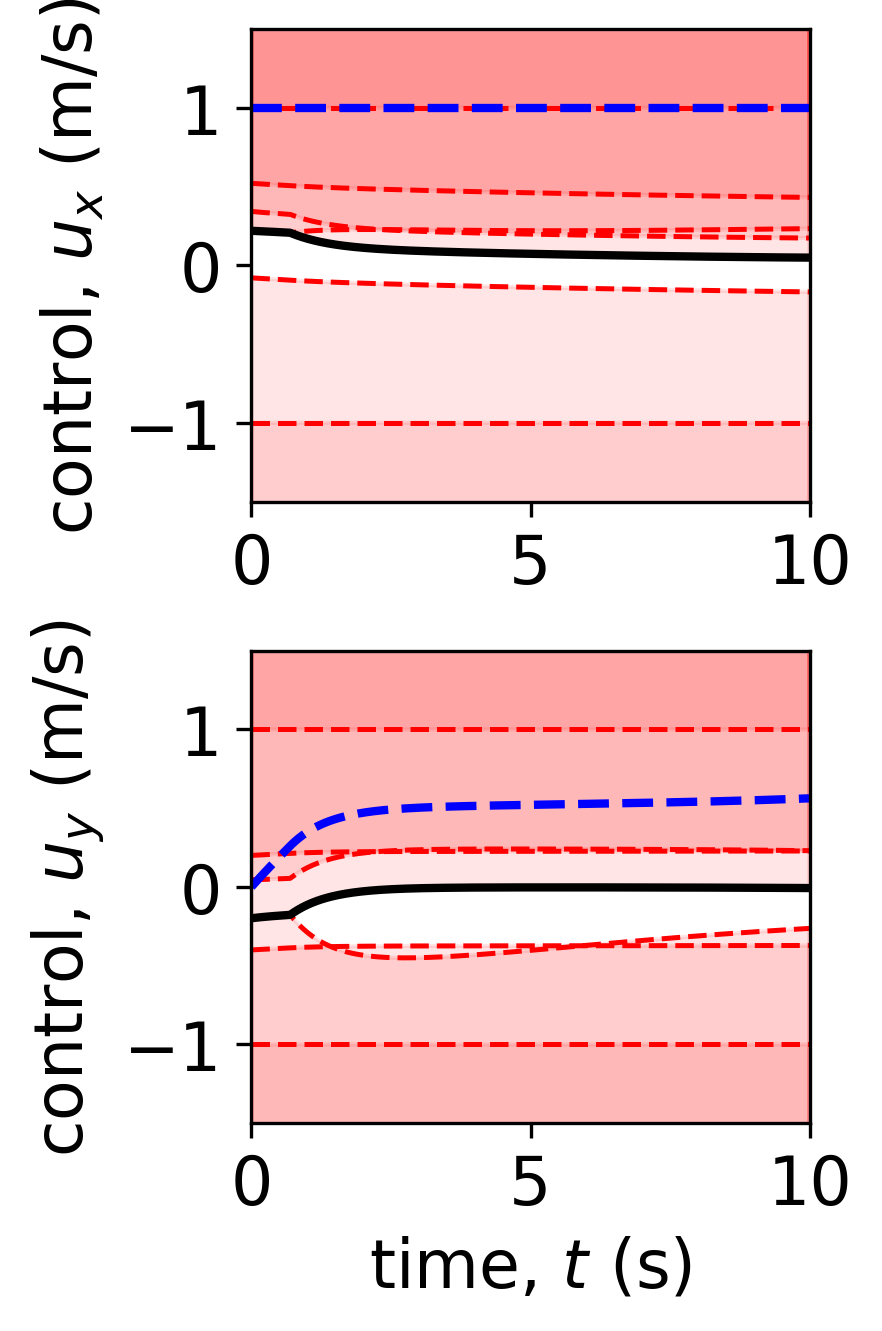}
        \caption[]{\small QP ($\omega=0.1$)}
        \label{fig:single_int_qp_0p1_u}
    \end{subfigure}\hfill
    \begin{subfigure}[t]{0.14\textwidth}
        \centering
        \includegraphics[width=\linewidth,trim=0mm 5mm 0mm 0mm]{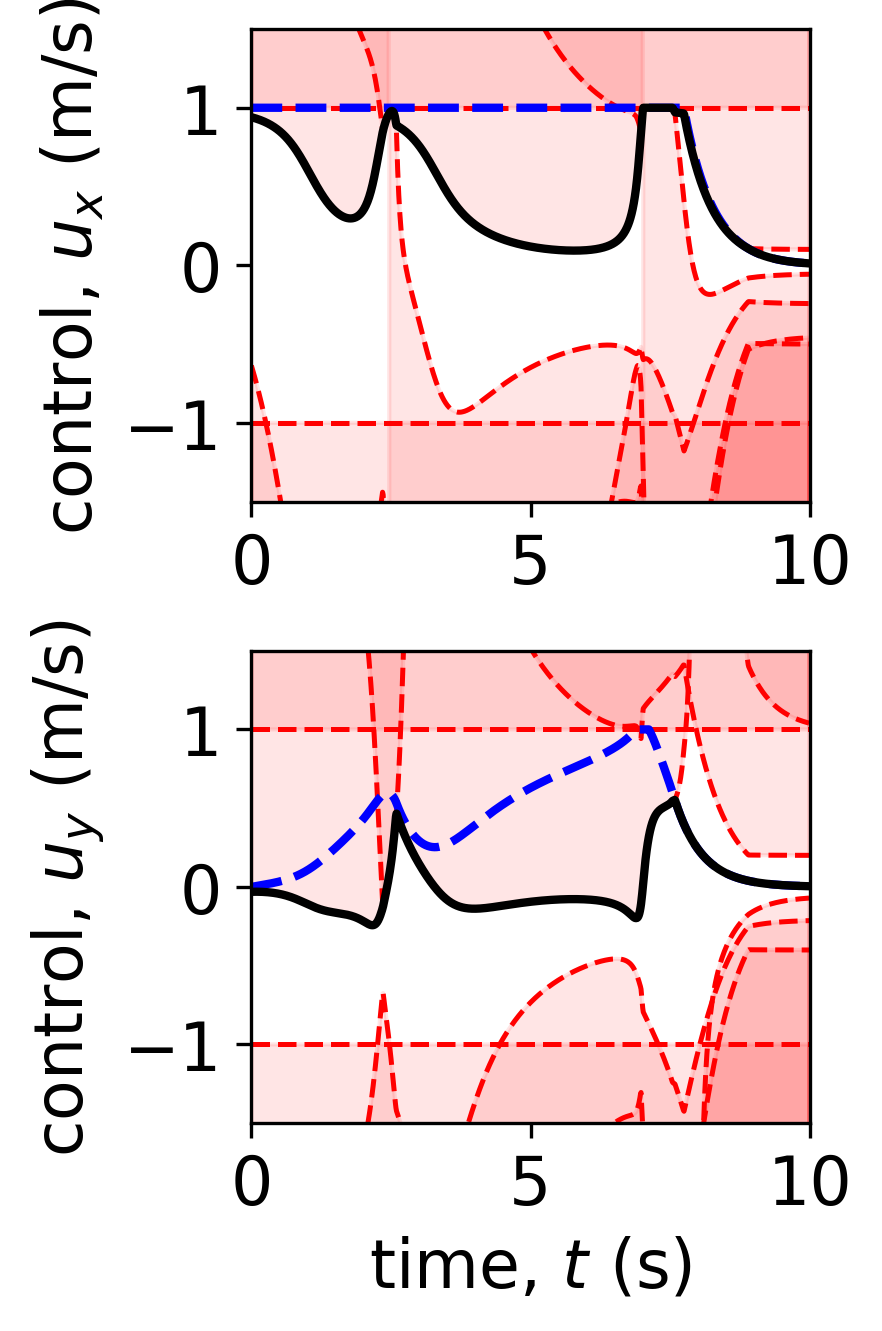}
        \caption[]{\small \!OD-QP \!\!($\omega\!=\!\!0.1$)\!}
        \label{fig:single_int_od_qp_0p1_u}
    \end{subfigure}\hfill
    \begin{subfigure}[t]{0.14\textwidth}
        \centering
        \includegraphics[width=\linewidth,trim=0mm 5mm 0mm 0mm]{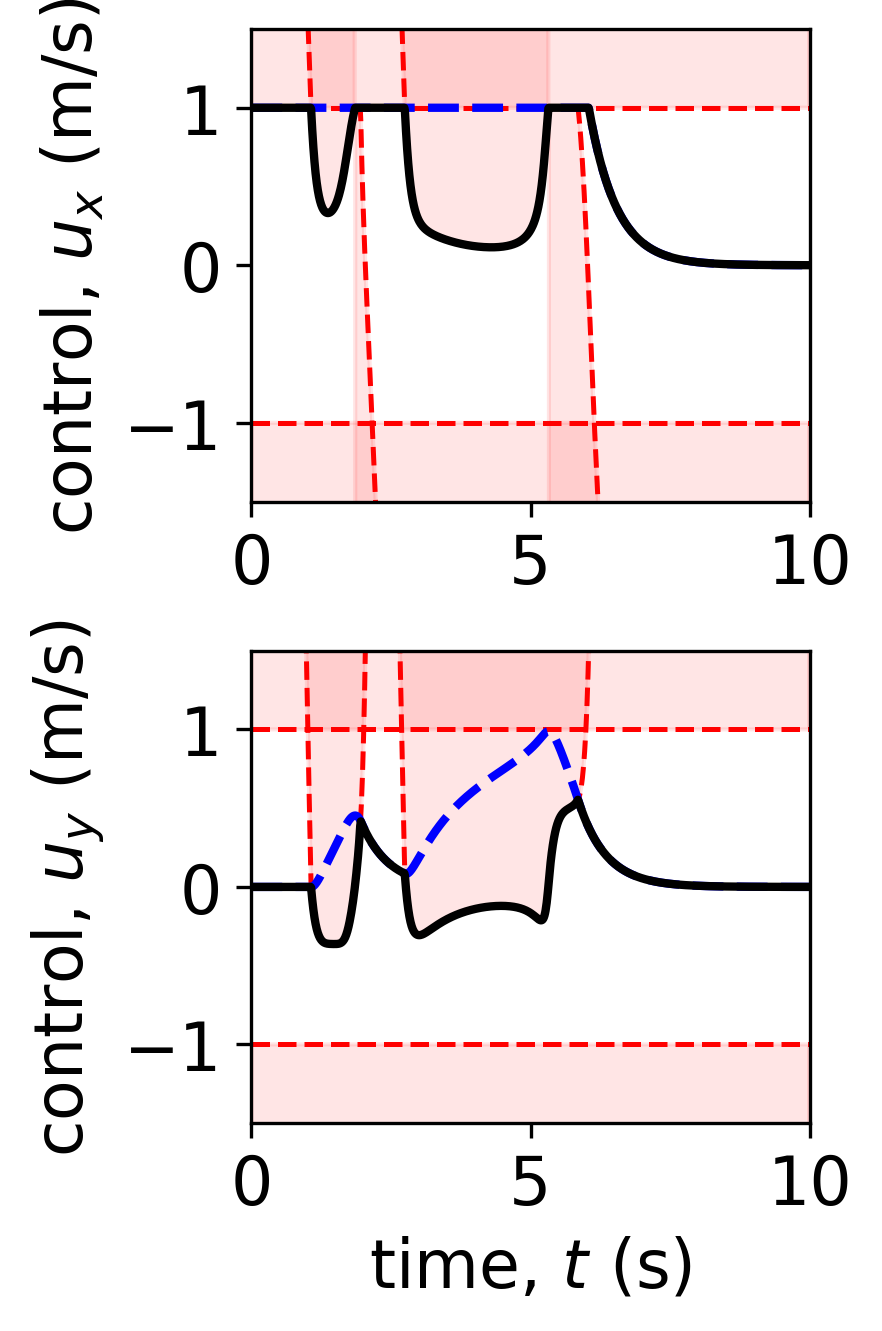}
        \caption[]{\small OD-QP \!($\omega\!=\!10$)}
        \label{fig:single_int_od_qp_10_u}
    \end{subfigure}\hfill
    \begin{subfigure}[t]{0.14\textwidth}
        \centering
        \includegraphics[width=\linewidth,trim=0mm 5mm 0mm 0mm]{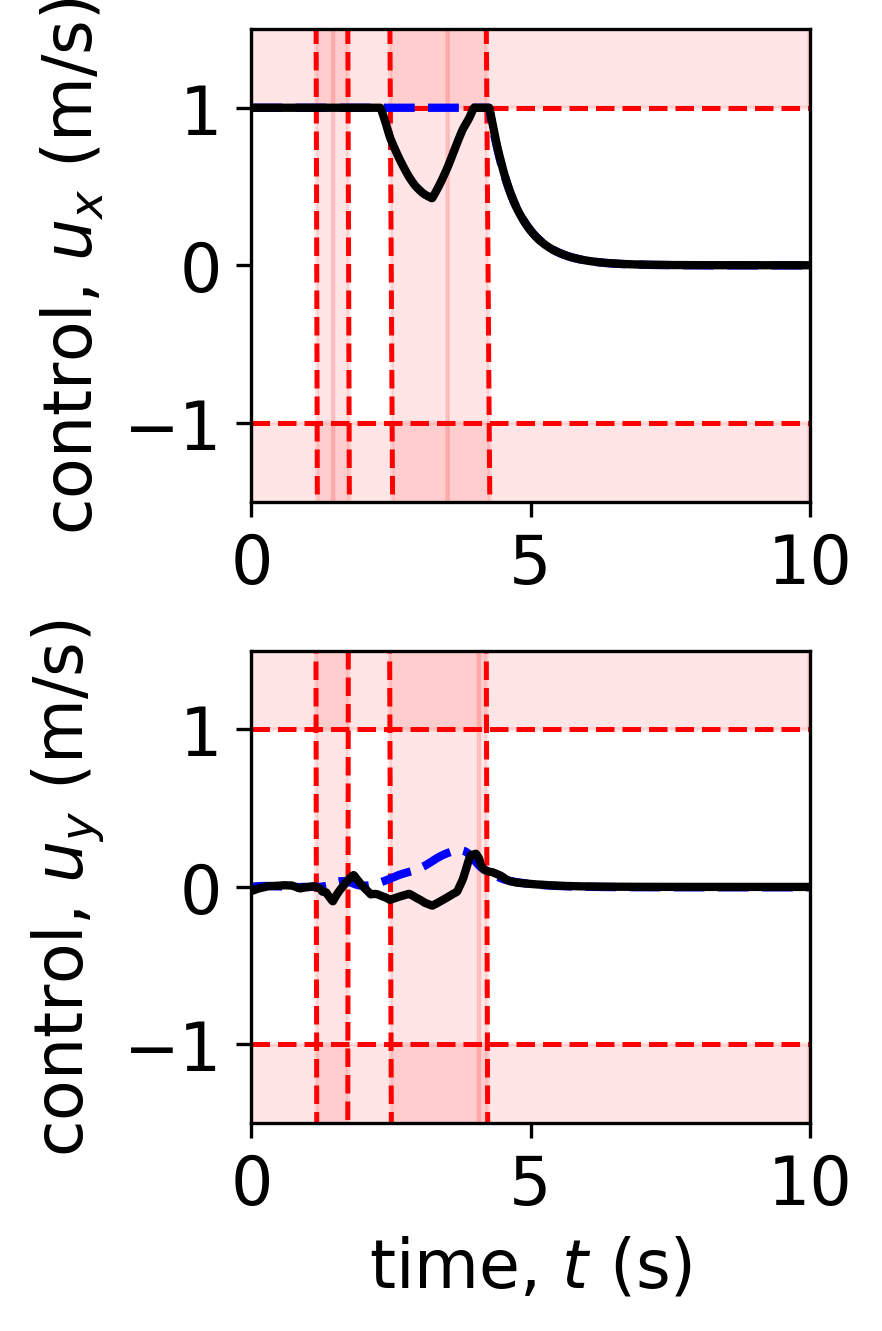}
        \caption[]{\small NN ($\tilde{\alpha}_\vartheta$)}
        \label{fig:single_int_nn_alpha_u}
    \end{subfigure}\hfill
    \begin{subfigure}[t]{0.14\textwidth}
        \centering
        \includegraphics[width=\linewidth,trim=0mm 5mm 0mm 0mm]{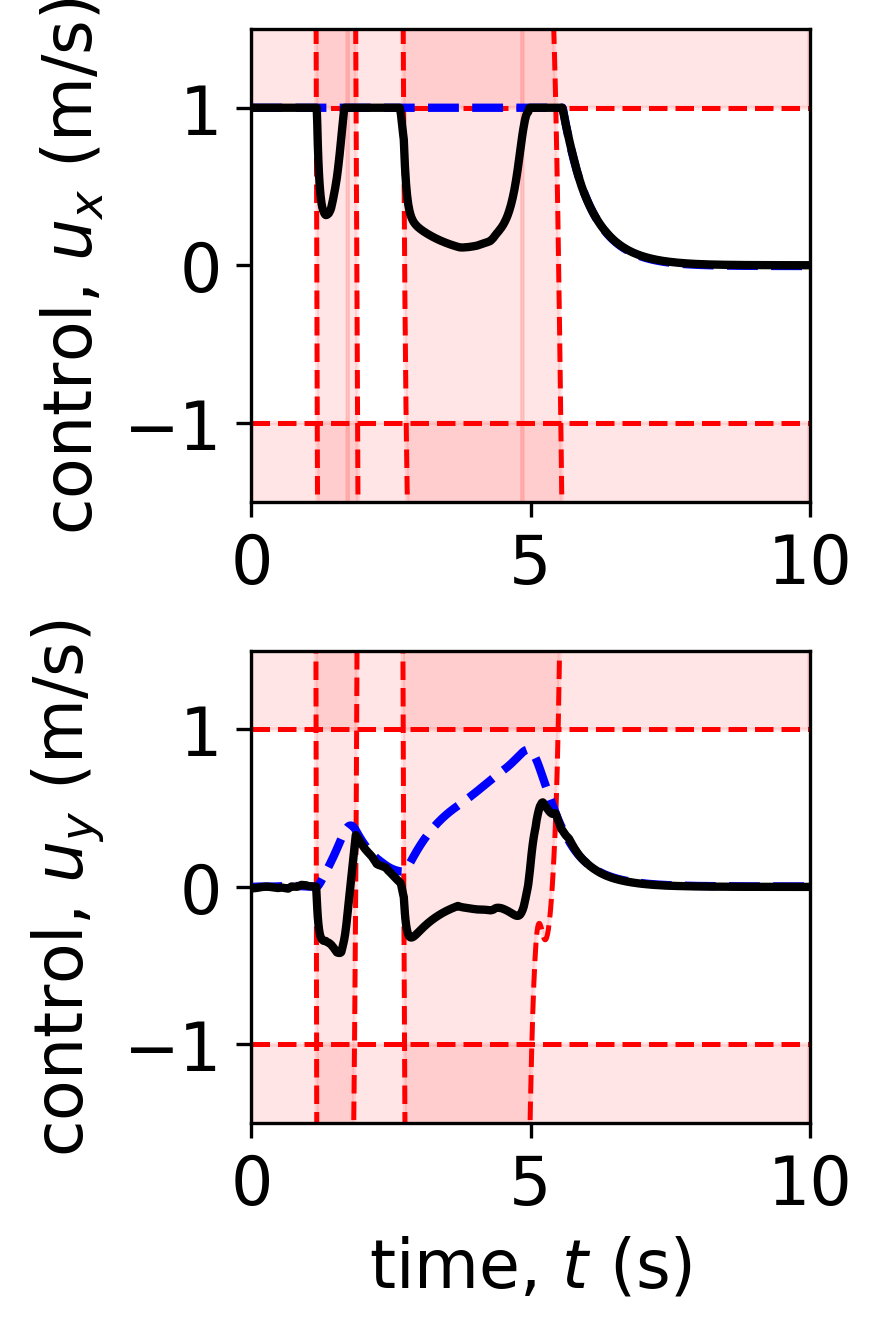}
        \caption[]{\small HardNet ($\tilde{\alpha}_\vartheta$)}
        \label{fig:single_int_hardnet_alpha_u}
    \end{subfigure}\hfill
    \begin{subfigure}[t]{0.14\textwidth}
        \centering
        \includegraphics[width=\linewidth,trim=0mm 5mm 0mm 0mm]{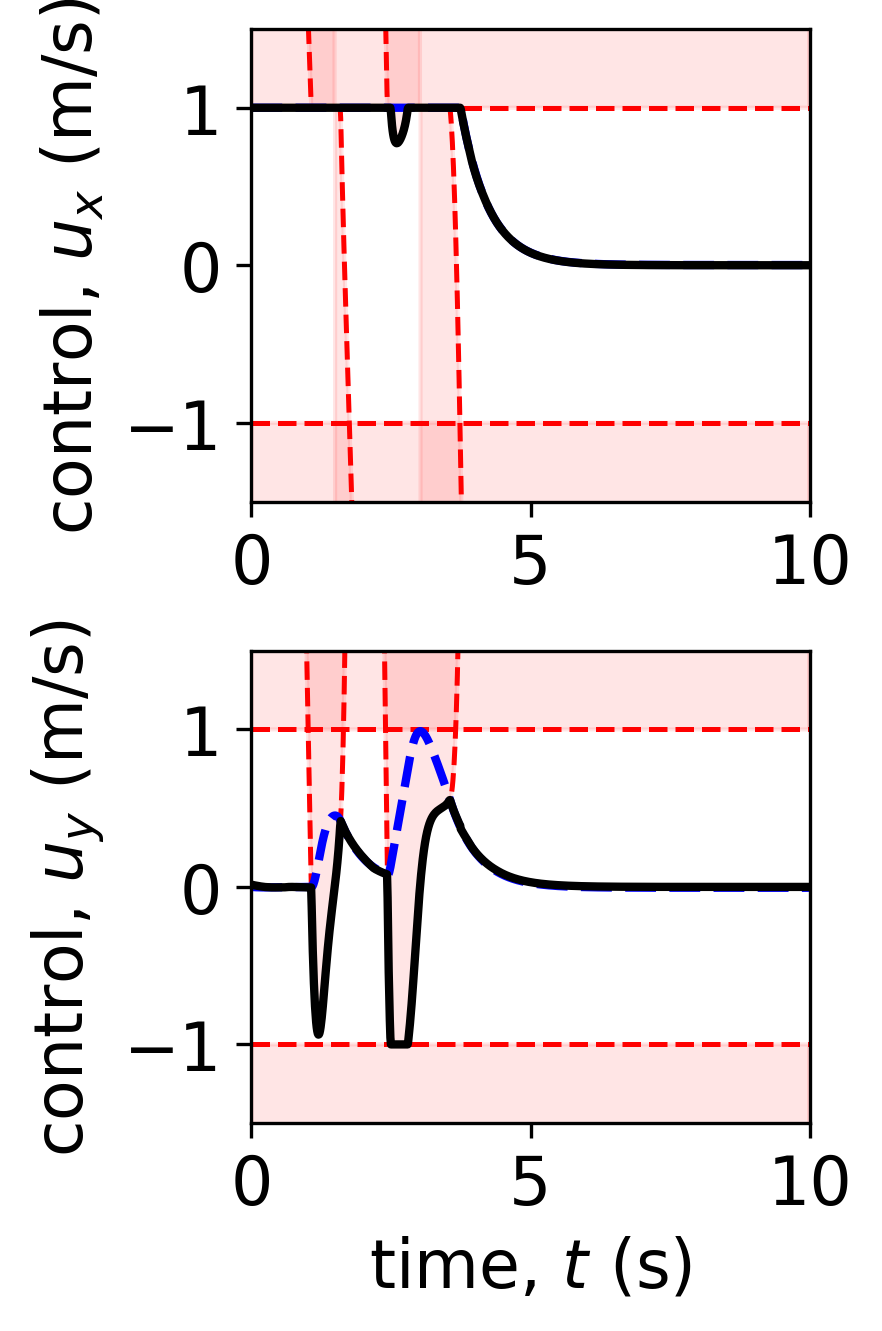}
        \caption[]{\small CAffNet ($\tilde{\alpha}_\vartheta$)}
        \label{fig:single_int_caffnet_alpha_u}
    \end{subfigure}\hfill
    \begin{subfigure}[t]{0.14\textwidth}
        \centering
        \includegraphics[width=\linewidth,trim=0mm 5mm 0mm 0mm]{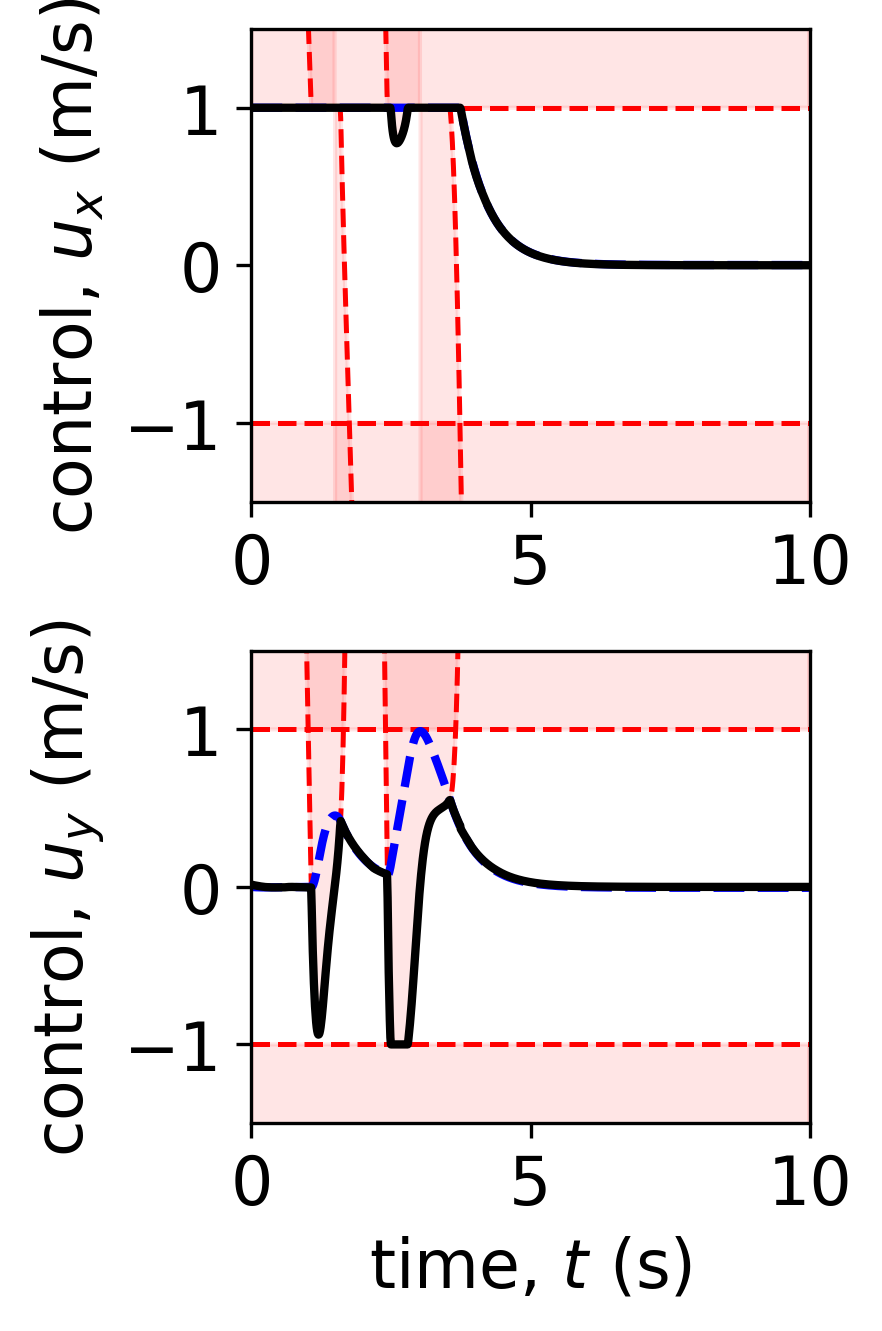}
        \caption[]{\small CAffNet-L.\! \!($\tilde{\alpha}_\vartheta$)}
        \label{fig:single_int_caffnet_lite_alpha_u}
    \end{subfigure}
    \caption[]{Control inputs for the single-integrator system under QP, optimal-decay QP (OD-QP), and neural-network-based methods. Parentheses indicate the CBF parameter used in each method: fixed $\omega=0.1$ or $\omega=10$, or learned $\tilde{\alpha}_\vartheta$. The (unsafe) nominal control inputs are depicted with dashed blue lines while the modified control inputs from the various methods are depicted with solid black lines. The red shaded regions depict unsafe control inputs corresponding to input bounds and the resulting  inputs that violate CBF conditions.}\vspace{-0.6cm}
    \label{fig:single_int_control_inputs}
\end{figure*}

\begin{figure}[t]
\vspace{0.2cm}
  \begin{minipage}{0.48\linewidth}
    \centering
    \includegraphics[width=\linewidth]{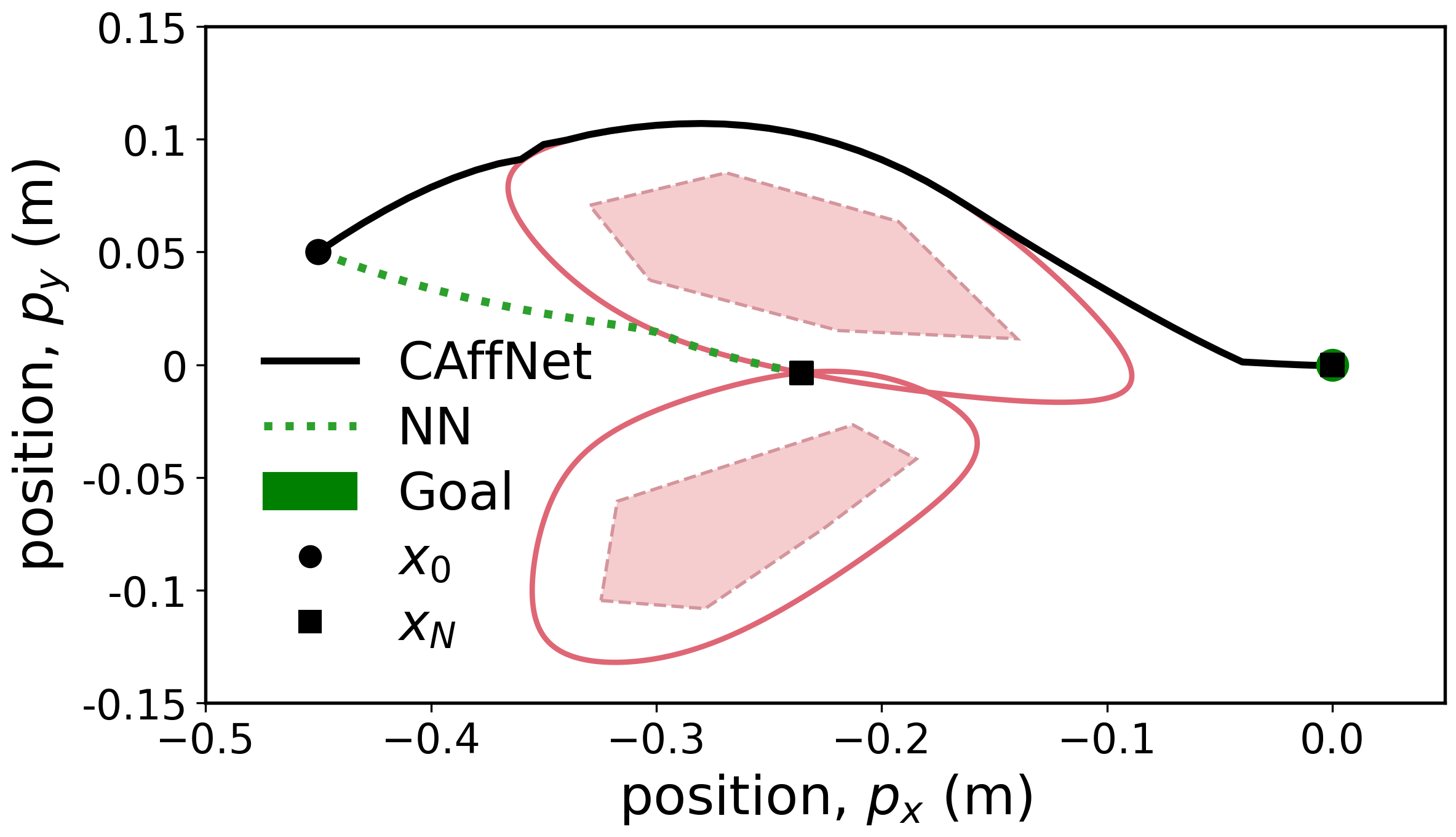}
    \caption{\small Comparison of CAffNet with an NN. CAffNet avoids local minima while the NN (with post-hoc safety filter) is trapped.}
    \label{fig:left_image}
  \end{minipage}
  \hfill
  \begin{minipage}{0.48\linewidth}
    \centering

\includegraphics[width=\linewidth,trim=0mm 0mm 20mm 0mm,clip]{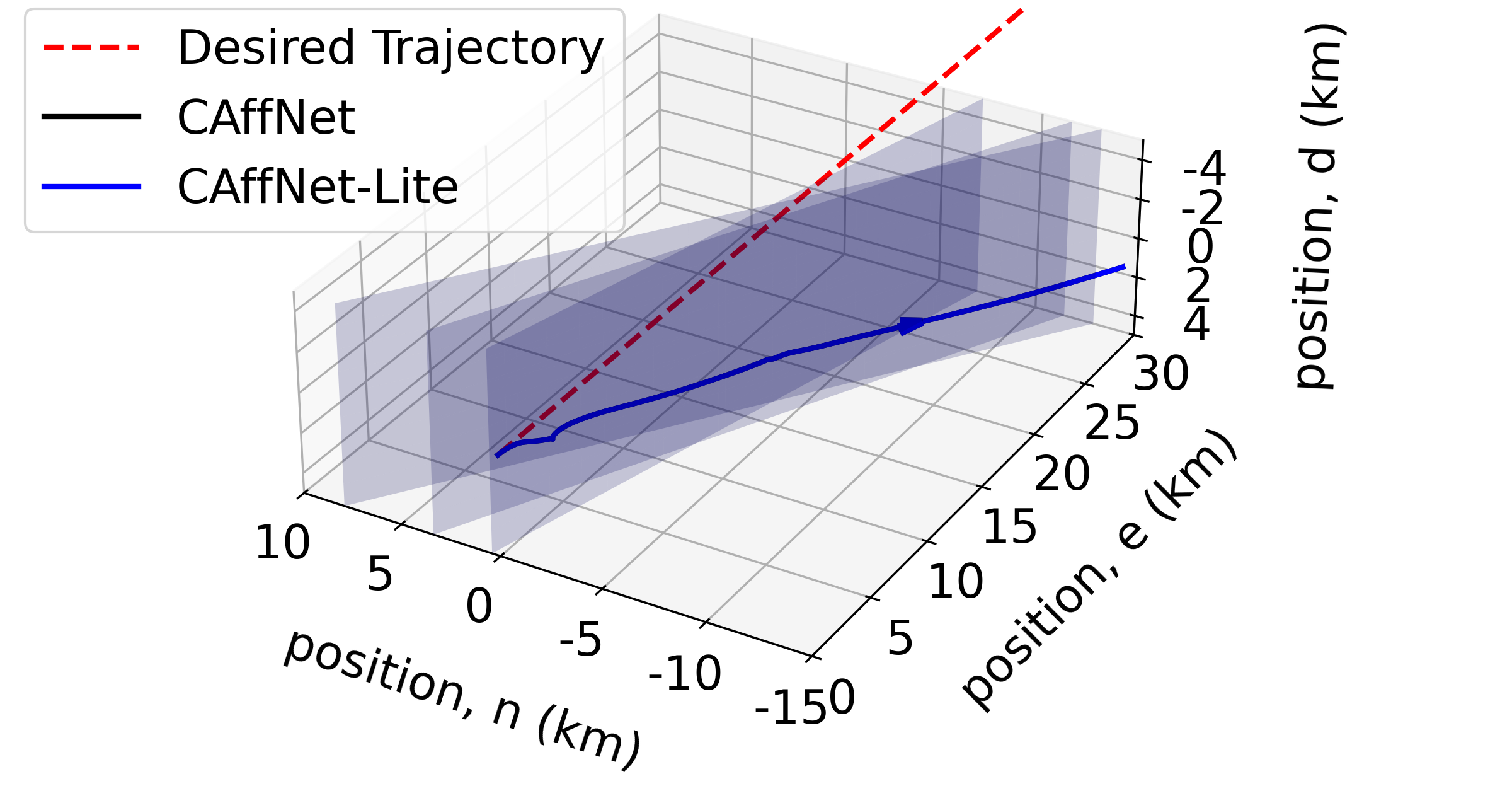}%
    \caption[]{\small Trajectories of CAffNet and CAffNet-Lite for the aircraft system, where both successfully avoided the geofences.}
    \label{fig:aircraft_all}
     \end{minipage}\vspace{-0.2cm}
\end{figure}

First, we consider a single integrator system ($\dot{x} = u$),
where the state $x = [p_x, p_y]^T \in \mathbb{R}^2$ denotes the position, and $u=[v_x, v_y]^T \in \mathbb{R}^2$ represents the velocities in the $x$ and $y$ directions. The state constraints are $p_x \in [-5, 1]$ and $p_y \in [-4, 2]$, which can be enforced by $h_x(x) = b_x - A_x x \ge 0$. The control constraints are $v_x, v_y \in [-1, 1]$. In addition, we consider three convex polyhedral obstacles with the hyperplane representation, $\mathcal{O}_j = \{x \in \mathbb{R}^2 \mid A_j x \le b_j\}, j \in \{1,2,3\}$. For each obstacle $\mathcal{O}_j$, the safety condition for its $i$-th edge is $h_j^i = a_j^i x - b_j^i \geq 0, \forall i \in \{1,\dots,n_j\}$, where $n_j$ denotes the number of edges. To enforce safety satisfaction for all edges of each obstacle, we use the smooth union function in \cite{molnar2023composing} to combine them into a single CBF, i.e.,
    $h_j(x) = \tfrac{1}{\kappa} \ln \left( \sum_{i=1}^{n_j} e^{\kappa h_j^i(x)} \right) - \tfrac{\ln n_j}{\kappa}$,
where $\kappa=10$ is a positive parameter that controls the smoothness of the union function. The corresponding Lie derivatives for $h_j(x)$ are $L_f h_j(x) = \sum_{i=1}^{n_j} \lambda_j^i(x) L_f h_j^i(x)$ and $L_g h_j(x) = \sum_{i=1}^{n_j} \lambda_j^i(x) L_g h_j^i(x)$, with $\lambda_j^i(x) = \mathrm{e}^{\kappa (h_j^i(x) - h_j(x))}$.
We assume that the composed smooth CBF remains a valid CBF for the corresponding obstacle-avoidance safe set.
The safe set for obstacle avoidance is defined as $\mathcal{C} = \{x \in \mathbb{R}^2 \mid h_j(x) \geq 0, j \in \{1,2,3\}\}$, which can be enforced by
\begin{align*}
        L_f  h_j(x)+ L_g h_j(x) u(x)
    \ge  -\alpha(h_j(x)).
\end{align*}
\indent In this simulation,
we compare the performance of QP,
optimal-decay QP (OD-QP), NN, HardNet, CAffNet, and CAffNet-Lite.
For QP and OD-QP, we set $\alpha(h(x)) = \omega h(x)$ with $\omega \in \{0.1,10\}$ and for NN, HardNet, CAffNet, and CAffNet-Lite, we compare both the case with $\alpha(h(x)) = \omega h(x)$ with $\omega \in \{0.1,10\}$ and the case
with a neural network parameterization $\alpha_\vartheta(x, h(x)) =   \tilde{\alpha}_\vartheta(x) h(x)$.
Each controller network uses ReLU activation functions and three hidden layers with $200$ neurons per layer. The neural network used in $\tilde{\alpha}_\vartheta(x)$ is a fully connected network with three hidden layers of 64 neurons each. The training samples are randomly generated within the safe set, and the nominal command is
the unsafe proportional control $u_{nom} = k_p(x_{ref} - x)$, where $k_p = 2$, and $x_{ref} = [0, 0]^T$ is the goal position.
Moreover, nominal and actual control inputs are saturated by the control constraints.
We set the training epochs to 10000 in this simulation, and trajectories with $T_{test} = 10s$ and $dt = 0.01s$ are shown in Fig.~\ref{fig:CBF_trajectory}.

With a fixed poorly tuned CBF parameter (Fig.~\ref{fig:CBF_trajectory}(a)), all methods, except OD-QP that has an adaptable parameter, are unable to achieve the goal position due to the conservativeness of the chosen parameter. On the other hand, with a good choice of CBF parameter (Fig.~\ref{fig:CBF_trajectory}(b)), all methods can reach the goal position, but NN and HardNet violate the safe state constraints. Moreover, with jointly learned CBF parameter $\alpha_\vartheta(x,h(x))$, CAffNet and CAffNet-Lite can successfully arrive at the goal region, while strictly satisfying all safety constraints during the entire trajectory. Furthermore, from Fig.~\ref{fig:single_int_control_inputs}, we can observe that NN and HardNet, even with the learned CBF parameter $\alpha_\vartheta(x,h(x))$, violate safe input constraints, whereas OD-QP is observed to require more frequent safety interventions (i.e., deviations from the nominal controller) than CAffNet and CAffNet-Lite. The same observation can be gleaned from
Table~\ref{CBF-table}, where CAffNet and CAffNet-Lite achieve the lowest mean cost, even lower than with a relatively well-tuned CBF parameter, while strictly satisfying all constraints and maintaining low computational cost.
This improvement occurs because the learned CBF term allows CAffNet to choose less conservative controls while still satisfying the safety constraints, so the trajectory reaches the goal earlier and stays closer to the unsafe nominal controller.
Both NN and HardNet have some constraint violations, and QP with fixed CBF parameters has a relatively higher cost and
a longer testing time.
The result shows that system behavior may be unexpected if the CBF parameter is not well designed. However, with the proposed framework, the CBF parameter can be learned together with the control policy, which can significantly improve the performance while strictly satisfying all safety constraints.

Further, Fig.~\ref{fig:left_image} shows a representative example where CAffNet was aware of the obstacle constraints and avoided the gap between them, whereas NN with soft constraints was unaware. Even with a post-hoc safety filter, it became active too late to prevent getting stuck. Empirically, NN with a CBF safety filter becomes stuck more frequently than CAffNet.
\vspace{-0.2cm}

\subsection{Geofencing of a Fixed-Wing Aircraft}
\vspace{-0.05cm}

In this example, we consider a fixed-wing aircraft model and control in \cite{molnar2025collision}, which has a higher control input dimension, to show the efficiency improvement of CAffNet-Lite compared to the original CAffNet. The state of the system is defined as $x = [n, e, d, \phi, \theta, \psi, V_T]^\top \in \mathbb{R}^7$, where $n$, $e$, and $d$ represent the north, east, and down positions, $\phi$, $\theta$, and $\psi$ are the roll, pitch, and yaw angles, and $V_T$ is the speed. The control input is defined as $u = [A_T, P, Q]^\top \in \mathbb{R}^3$, where $A_T$ is the acceleration, $P$  the roll rate, and $Q$  the pitch rate.

In this example, the safety condition is constructed by enforcing geofences $h_{p,i}(r) \ge 0, \forall t \ge 0$, where
\begin{align*}
    h_{p,i}(r) = n_i^\top (r - r_i) - \rho_i, \quad \dot{h}_{p,i}(r) = n_i^\top v,
\end{align*}
with a center $r_i$ and a normal vector $n_i$. For brevity, the readers are referred to \cite{molnar2025collision} for
the system dynamics
~\cite[Eq. (8)]{molnar2025collision}, (smoothed) CBF condition
~\cite[Eq. (47)]{molnar2025collision}
, and other details.
To train the neural network, we use the $v_s$
in
\cite[Eq. (48)]{molnar2025collision} as the nominal controller, and train CAffNet and CAffNet-Lite to learn this nominal command with MSE loss. Both networks use ReLU activation functions and 16 hidden layers with 256 neurons per layer. Note that the nominal controller is computed with the smooth function; however, in CAffNet and CAffNet-Lite, we use the original CBF conditions without approximating them. The training samples are generated by randomly initializing 10 safe states and simulating with the nominal controller for $T_{test}=300s$. We trained both networks for 1000 epochs.
From the trajectories in Fig.~\ref{fig:aircraft_all} and velocity command in Fig.~\ref{fig:aircraft_u}, both methods deviate the aircraft from the desired trajectory to avoid geofences. The evaluation result is shown in Table~\ref{tab:aircraft_results}, where CAffNet-Lite uses lower training and testing time than CAffNet, which demonstrates the efficiency improvement of CAffNet-Lite compared to the original CAffNet, making it more suitable for real-time applications and higher-dimensional systems.
However, due to more projection candidates in CAffNet, it converges faster than CAffNet-Lite under the same training epoch, leading to a lower mean cost.

\begin{figure}[t]
    \centering
    \vspace{0.15cm}
    \begin{subfigure}[t]{0.49\textwidth}
        \centering
        \includegraphics[width=\linewidth,trim=5mm 3mm 0mm 0mm]{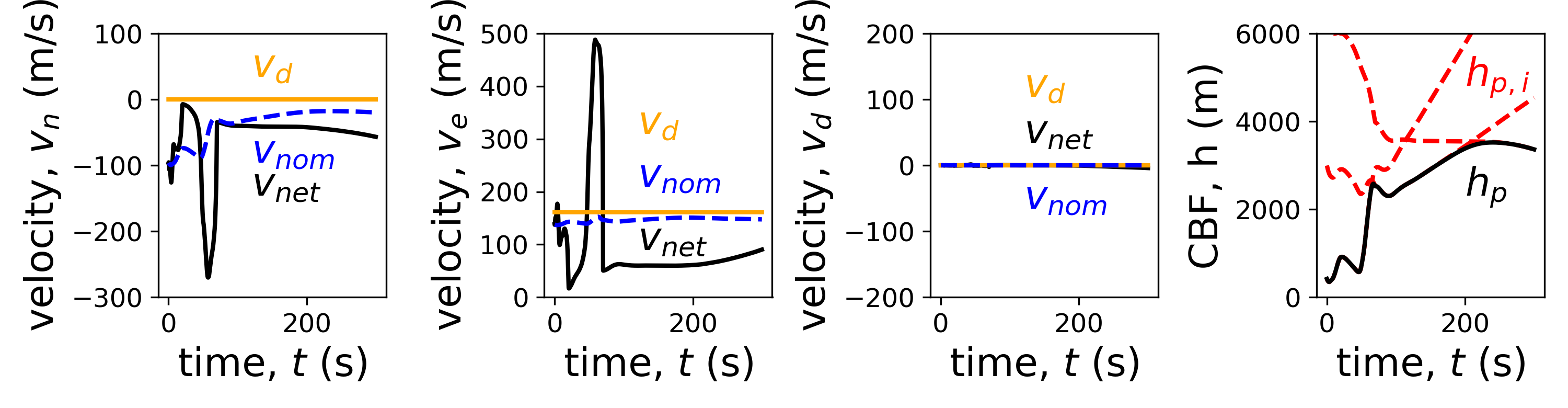}
        \caption[]{\small Safe velocity command and CBFs of CAffNet.}
        \label{fig:aircraft_CAff_u}
    \end{subfigure}\vfill
    \begin{subfigure}[t]{0.49\textwidth}
        \centering
        \includegraphics[width=\linewidth,trim=5mm 3mm 0mm 0mm]{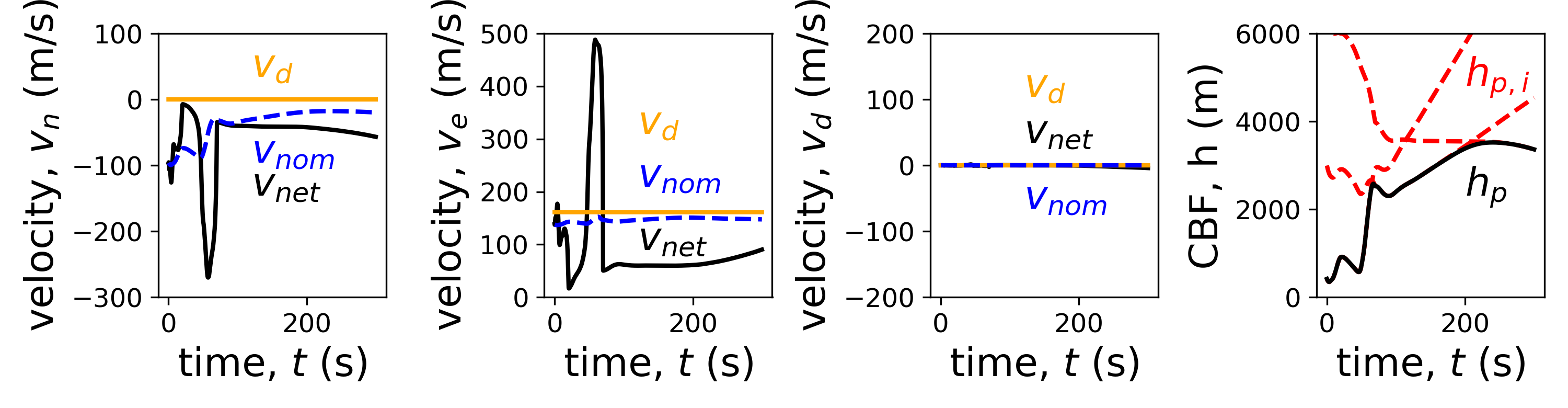}
        \caption[]{\small Safe velocity command and CBFs of CAffNet-Lite.}\vspace{-0.1cm}
        \label{fig:aircraft_LiteCAff_u}
    \end{subfigure}\hfill
    \caption[]{Safe velocity command and CBF values of CAffNet and CAffNet-Lite for the aircraft system. Both methods can successfully learn the safe velocity command that satisfies all CBF conditions.}
    \label{fig:aircraft_u}\vspace{-0.2cm}
\end{figure}

\begin{table}[t]
  \caption{Results of learning safe control policies for the aircraft system with parenthetical values  representing standard deviations. Both CAffNet and CAffNet-Lite satisfy all constraints. CAffNet-Lite has lower training and testing time than CAffNet.}
  \label{tab:aircraft_results}\vspace{-0.2cm}
  \begin{center}
    \begin{small}
      \begin{sc}
\resizebox{\linewidth}{!}{%
\begin{tabular}{lcccccc}
\toprule
Method & Cost & \multicolumn{3}{c}{Inequality violation} & $T_{train}$ & $T_{test}$ \\
\cmidrule(lr){3-5}
 &  & Max & Mean & Num (\%) & (ms) & (s) \\
\midrule
CAffNet-FF & $\mathbf{3.82 \times 10^7}$ & \textbf{0.00} & \textbf{0.00} & \textbf{0.00\%} & 118.67 & 25.39 \\
 & $(1.37 \times 10^7)$ & (0.00) & (0.00) & (0.00\%) & (1.12) & (0.50) \\
CAffNet-Lite & $\mathbf{3.82 \times 10^7}$ & \textbf{0.00} & \textbf{0.00} & \textbf{0.00\%} & \textbf{85.52} & \textbf{21.07} \\
 & $(1.37 \times 10^7)$ & (0.00) & (0.00) & (0.00\%) & (0.65) & (0.30) \\
\bottomrule
\end{tabular}
}
      \end{sc}
    \end{small}
  \end{center}
\end{table}

\vspace{-0.2cm}
\section{Conclusion}
\vspace{-0.05cm}
We presented CAffNet-Lite, a safe-by-design neural control framework that jointly learns a neural network controller alongside neural-network-parameterized CBF conditions. The proposed method guarantees hard satisfaction of control-affine constraints and strong forward invariance of the safe set by construction, without requiring a separate QP-based safety filter. A lightweight constraint decomposition strategy improves scalability over the prior CAffNet approach while preserving feasibility and universal approximation guarantees. Future work will consider extending the framework to broader classes of constraints and evaluating performance on real-world systems.\vspace{-0.1cm}

\bibliographystyle{IEEEtran}
\bibliography{references}

\end{document}